\documentclass[aps,pra,superscriptaddress,twocolumn]{revtex4-1}
\usepackage[utf8]{inputenc}
\usepackage[english]{babel}
\usepackage{amsmath}
\usepackage{amsfonts}
\usepackage{amssymb}
\usepackage{graphicx}
\usepackage{amsthm}
\usepackage{color}
\usepackage{hyperref}

\newcommand{\norm}[1]{\ensuremath{\left|\left|#1\right|\right|}}

\begin{document}
\title{Optimal quantum state discrimination via nested binary measurements
} 
\author{Matteo Rosati}
\affiliation{NEST, Scuola Normale Superiore and Istituto Nanoscienze-CNR, I-56127 Pisa,
Italy.}
\author{Giacomo De Palma} 
\affiliation{QMATH, Department of Mathematical Sciences, University of Copenhagen,
Universitetsparken 5, 2100 Copenhagen, Denmark}
\author{Andrea Mari} 
\affiliation{NEST, Scuola Normale Superiore and Istituto Nanoscienze-CNR, I-56127 Pisa,
Italy.}
\author{Vittorio Giovannetti} 
\affiliation{NEST, Scuola Normale Superiore and Istituto Nanoscienze-CNR, I-56127 Pisa,
Italy.}

\begin{abstract} A method to compute the optimal success probability of discrimination of $N$ arbitrary quantum states is presented, based on the decomposition of any $N$-outcome measurement into sequences of nested two-outcome ones. In this way the optimization of the measurement operators can be carried out in successive steps, optimizing first the binary measurements at the deepest nesting level and then moving on to those at higher levels. We obtain an analytical expression for the maximum success probability after the first optimization step and examine its form for the specific case of $N=3,4$ states of a qubit. In this case, at variance with previous proposals, we are able to provide a compact expression for the success probability of any set of states, whose numerical optimization is straightforward; the results thus obtained highlight some lesser-known features of the discrimination problem.  
\end{abstract}
\maketitle

\section{Introduction}
The discrimination of quantum states~\cite{QSDiscriminationRev} is one of the fundamental problems in Quantum Information and a basic task for several applications in communication~\cite{HolBound,hauswoot,seq1,polarWildeGuha,NOSTROHol}, cryptography~\cite{BB84,CryptoRev,MinMaxEnt}, fundamental questions~\cite{NoCloning,AsyCloning,NoSignal,DimWitness}, measurement and control~\cite{MilburnBOOK,QIllumination} and algorithms~\cite{HidSubgroup}. Triggered by the observation that non-orthogonal quantum states cannot be perfectly discriminated, this subject has stimulated much work, both from a theoretical and practical point of view: the seminal works of Helstrom~\cite{Hel}, Holevo~\cite{HolOp} and Yuen \textit{et al.}~\cite{YKL} formalized the problem, obtaining a set of conditions for the optimal measurement operators, which in turn provide the optimal success probability, then solved it for sets of states symmetric under a unitary transformation; more recently, acknowledging that a general analytical solution is hard to find, research focused on finding a solution for sets with more general symmetries~\cite{geoUni,Chiribella,guhaGroup}, computing explicitly the optimal measurements for the most interesting sets of states~\cite{multiHel1,multiHel2,bae,threeQubits,qubits} and studying the implementation of such measurements with available technology (see for example ~\cite{Ken,OpGauss,OpKen,NOSTRODisc,Marq,ImplProj,AdaptiveLOCC,OpMeasCoh,Dol,QKD,SeqNull,DolExp,becerra1} for the case of two optical coherent states, the most relevant for optical communication). Also, the problem of discrimination has been identified as a convex optimization one, arguing that it can be solved efficiently with numerical optimization methods~\cite{SDP}.
\\ In this article we attempt to solve the optimal discrimination of $N$ quantum states from a different perspective, by providing a structured expression for the $N$-outcome Positive Operator Valued Measure (POVM) used to discriminate the states. Indeed it can be shown~\cite{prec,precMap} that any POVM comprising $N$ elements is equivalent to a collection of binary POVM's, i.e., comprising two elements, as the one employed in Ref.~\cite{NOSTROHol}: depending on the binary outcome of the first measurement, a second one is applied; its binary outcome in turn affects the choice of the third binary measurement and so on. In this way a sequence of nested binary POVM's can be constructed, where the POVM applied at a given level depends on the string of binary outcomes of the previous ones. 
This result was already obtained in Ref.~\cite{prec}. When applied to state discrimination, it acquires a more operational meaning: each binary POVM can be seen as discriminating between two subsets of the initial set of states, identified by previous outcomes. Hence the sequence of measurements induces a sequence of discrimination probabilities, so that, if the optimization problem is solved independently for any set of a fixed number of states, the result can be employed in the optimization problem for larger sets of states. \\
In the second part of the article, employing this decomposition and the two-state optimal probability~\cite{Hel}, we obtain an expression for the success probability of discrimination of any $N=3,4$ states, depending on a single measurement operator, and solve the problem analytically for specific sets of states. Then we restrict our attention to qubit states and obtain a compact expression which can be easily optimized numerically case by case, at variance with less compact results for $N=3$ presented in previous works based on Bloch-space geometry~\cite{qubits,bae,threeQubits}. We recover the results of those works and highlight in particular some interesting lesser-known implications of Ref.~\cite{bae}.\\
The article is structured as follows: in Sec.~\ref{Deco} we describe the decomposition in terms of nested binary POVM's and  provide a proof of its validity, similar to that of Ref.~\cite{prec}; in Sec.~\ref{StateDisc} we apply it to state discrimination and obtain an explicit expression for the case of $N=3,4$ arbitrary states, then discuss its optimization in some specific cases; in Sec.~\ref{Qubits} we treat the case of $N=3,4$ qubit states, computing a compact expression which can be optimized numerically and highlighting some results obtained in this way. Eventually in Sec.~\ref{Conc} we draw some conclusions. Detailed computations of the quantities appearing in the article are provided in the Appendices.

\section{General decomposition of a $N$-outcome measurement into nested binary ones}\label{Deco}
In this Section we prove that any quantum measurement with an arbitrary number of outcomes can be decomposed into a sequence of nested measurements with binary outcomes, where the previous results determine the choice of successive measurements. We stress that the same result was obtained in Ref.~\cite{prec}. At variance with the latter, our proof does not make use of the spectral decomposition of the initial measurement operators; we present it here in a form adapted to the main purpose of the article.  Let us suppose we want to perform a quantum measurement with $N$ possible outcomes: it can be expressed in general as a POVM $\mathcal{M}^{(N)}$ of elements $E_{j}$, one for each outcome $j=0,\cdots,N-1$, satisfying the positivity and completeness conditions, i.e., respectively $E_{j}\geq0$ and $\sum_{j=0}^{N-1}E_{j}=\mathbf{1}$, where $\mathbf{1}$ is the identity operator on the Hilbert space of the system to be measured. This expression can be interpreted as a one-shot measurement with several possible results and its practical realization may often be very hard. 
On the other hand we could restrict to performing only measurements with two outcomes, as described by \textit{binary} POVM's: $\mathcal{B}\equiv\mathcal{M}^{(2)}=\{B_{0},B_{1}\}$. This may be useful when limited technological capabilities or specific theoretical requirements constrain the  number of allowed outcomes and the complexity of our measurement. 
It is then natural to ask whether this smaller set of resources is sufficient to describe a general quantum measurement. We answer positively by showing that the more general $N$-outcome formalism can be broken up into several binary steps and interpreted as a sequence of nested POVM's with two outcomes, trading a \textit{one-shot}, \textit{multiple-outcome} measurement for a \textit{multiple-step}, \textit{yes-no} measurement.\\
The nested POVM can be expressed in terms of \textit{conditional binary} POVM's $\mathcal{B}_{\vec{k}}=\{B_{\vec{k},0},B_{\vec{k},1}\}$, each complete by itself, to be applied only if a specific string $\vec{k}$ of previous results is obtained. 
 For example for $N=4$ the nested POVM can be realized in two steps and written compactly as the collection of three binary POVM's: 
  $\mathcal{N}^{(4)}=\left\{\mathcal{B}^{(2)}_{0},\mathcal{B}^{(2)}_{1}\right\}\circ\left\{\mathcal{B}^{(1)}\right\}$, properly composed as follows and shown in Fig.~\ref{schema}. The measurement starts by applying the first-step binary POVM $\mathcal{B}^{(1)}=\left\{B^{(1)}_{k_{1}}\right\}_{k_{1}=0,1}$ then, depending on its outcome $k_{1}$, it selects $\mathcal{B}^{(2)}_{k_{1}}=\left\{B^{(2)}_{k_{1},k_{2}}\right\}_{k_{2}=0,1}$ among the two POVM's available in the second-step collection $\left\{\mathcal{B}^{(2)}_{0},\mathcal{B}^{(2)}_{1}\right\}$. Eventually the chosen second-step POVM is applied, receiving an outcome $k_{2}$. The total outcome is a string of two bits, i.e., $k_{1},k_{2}$, whose value identifies one of  four possible outcomes, as desired. Suppose now to apply this measurement on a state $\rho$ of some physical system: if the first-step outcome is $k_{1}=0$, the resulting unnormalized evolved state is $\sqrt{B^{(1)}_{0}}\rho\sqrt{B^{(1)}_{0}}$; if then the second-step outcome is $k_{2}=0$, the final unnormalized state of the system is $\sqrt{B^{(2)}_{0,0}}\sqrt{B^{(1)}_{0}}\rho\sqrt{B^{(1)}_{0}}\sqrt{B^{(2)}_{0,0}}$. This means that the nested POVM has a more explicit representation as 
 \begin{eqnarray} \mathcal{N}^{(4)}=\left\{F_{k_{1},k_{2}}=\left|\sqrt{B^{(2)}_{k_{1},k_{2}}}\sqrt{B^{(1)}_{k_{1}}}\right|^{2}\right\}_{k_{1},k_{2}=0,1},\end{eqnarray} 
 where $|X|^{2}=X^{\dag}X$ is the square of the absolute value of an operator $X$. 
 \begin{figure}[t!]
 \includegraphics[scale=.24]{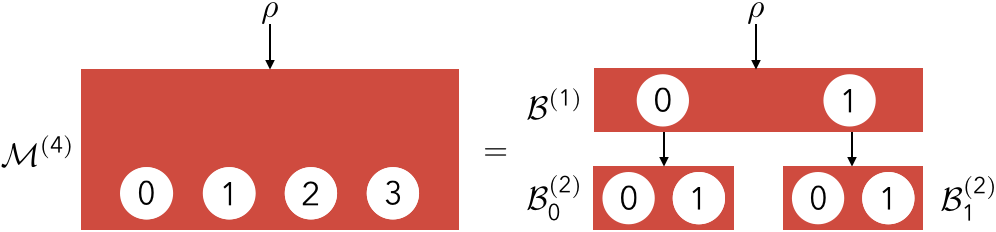}\caption{Schematic depiction of the nested decomposition for $N=4$, explicitly discussed in the text. Any four-outcome measurement $\mathcal{M}^{(4)}$ acting on a state $\rho$ is equivalent to the concatenation of two-outcome measurements: the first-step one $\mathcal{B}^{(1)}$, with result $k_{1}=0,1$, and the second-step ones $\mathcal{B}^{(2)}_{k_{1}}$, which are mutually exclusive and applied only if the corresponding first outcome $k_{1}$ was obtained. }
 \label{schema}
 \end{figure}
In the general case, let us indicate a sequence of $b-a+1$ bits as 
\begin{equation}\label{notation}
k_{(a,b)}=\begin{cases}
k_{a},k_{a+1},\cdots,k_{b},& b\geq a\\
\emptyset,&b<a, 
\end{cases}
\end{equation} 
and define as $\mathcal{B}^{(u)}_{k_{(1,u-1)}}=\left\{B^{(u)}_{k_{(1,u-1)},0},B^{(u)}_{k_{(1,u-1)},1}\right\}$ the binary POVM to be performed at the $u$-th step if the previous $u-1$ measurements had a sequence of results $k_{(1,u-1)}$. Then we can define a nested POVM $\mathcal{N}^{(N)}$ of order $N=2^{u_{F}}$ as 
\begin{eqnarray}
\mathcal{N}^{(N)}&=&\left\{\mathcal{B}^{(u_{F})}_{k_{(1,u_{F}-1)}}\right\}_{k_{1},\cdots,k_{u_{F}-1}=0,1}\circ\cdots\circ\left\{\mathcal{B}^{(1)}\right\}\nonumber\\
&=&\left\{F_{k_{(1,u_{F})}}=\left|\sqrt{B^{(u_{F})}_{k_{(1,u_{F})}}} \cdots\sqrt{B^{(1)}_{k_{1}}}\right|^{2}\right\},\label{nested}
\end{eqnarray}
i.e., the collection of $2^{u_{F}}-1$ binary POVM's $\mathcal{B}^{(u)}_{k_{(1,u-1)}}$, for all previous outcomes $k_{(1,u-1)}$ at a given step $u$ and all steps $u=1,\cdots,u_{F}$. We can certify that $\mathcal{N}^{(N)}$ so constructed actually is a POVM by checking positivity and completeness of its elements $F_{k_{(1,u_{F})}}$: the former requirement is trivial, while the latter follows from the fact that each binary POVM is complete, as shown in Appendix~\ref{appA}. 

 In light of the previous discussion we can now state the main theorem:
 \newtheorem{theorem}{Theorem}
\begin{theorem}\label{decomposition}
Any $N$-outcome POVM $\mathcal{M}^{(N)}=\{E_{j}\}_{j=0,\cdots,N-1}$ is equivalent to a nested POVM $\mathcal{N}^{(\tilde{N})}$, $\tilde{N}=2^{u_{F}}$, as in Eq.~\eqref{nested}, composed exclusively of binary POVM's $\mathcal{B}^{(u)}_{k_{(1,u-1)}}$, with a total number of steps $u_{F}$ equal to:
\begin{enumerate}
\item $\log_{2}N$, if $N$ is a power of $2$;
\item $\lceil\log_{2}N\rceil$ otherwise, where $\lceil\cdot\rceil$ is the ceiling function, equal to the smallest integer following the argument.
\end{enumerate}
 \end{theorem}
 \begin{proof}
 Consider the first case above, i.e., $N=2^{u_{F}}\equiv\tilde{N}$. 
We start by providing a binary representation of the labels $j$ of the initial POVM $\mathcal{M}^{(N)}$, i.e., we define $E_{k_{(1,u)}}\equiv E_{j^{(k)}}$, with $j^{(k)}=\sum_{u=1}^{u_{F}}2^{u-1}k_{u}$. In order to prove the theorem we have to show that by combining the elements of the initial $N$-outcome POVM $\mathcal{M}^{(N)}$ one can always define a set of binary POVM's $\mathcal{B}^{(u)}_{k_{(1,u-1)}}$, for all $k_{1},\cdots,k_{u-1}=0,1$ and $u=1,\cdots,u_{F}$, such that: i) their nested composition is a POVM of the form $\mathcal{N}^{(N)}$, Eq.~\eqref{nested}; ii) the elements of the latter are equal to the elements of $\mathcal{M}^{(N)}$.  \\
First of all we construct the binary elements at each step $u$, by taking the sum of all the elements $E_{k_{(1,u)},k_{(u+1,u_{F})}}$ with a fixed value of the first $u$ bits, then renormalizing it by all previous binary elements, as in a Square Root Measurement~\cite{schumawest,hauswoot}. For example define the elements of the first-step POVM $\mathcal{B}^{(1)}$ as
 \begin{eqnarray}\label{first}
 B^{(1)}_{k_{1}}=\sum_{k_{(2,u_{F})}}E_{k_{1},k_{(2,u_{F})}},
 \end{eqnarray}
 for each value of the outcome $k_{1}=0,1$. Being a sum of positive operators, the elements so defined are themselves positive; moreover their sum equals the sum of all the elements of $\mathcal{M}^{(N)}$, implying that they are complete. At the second step define the elements of the two possible POVM's $\mathcal{B}^{(2)}_{k_{1}}$ as
  \begin{eqnarray}\label{second}
 B^{(2)}_{k_{(1,2)}}=\sqrt{B^{(1)}_{k_{1}}}^{-1}\sum_{k_{(3,u_{F})}}E_{k_{(1,2)},k_{(3,u_{F})}}\sqrt{B^{(1)}_{k_{1}}}^{-1},
 \end{eqnarray}
 where the inverse of an operator is to be computed only on its support, while it is equal to $0$ on its kernel, i.e., its pseudo-inverse. 
 Also in this case the defined elements  are positive by construction, but they are not complete. Indeed it is easy to show, employing the definition \eqref{first}, that $B^{(2)}_{k_{1},0}+B^{(2)}_{k_{1},1}=\mathbf{1}_{k_{1}}$. Here $\mathbf{1}_{k_{1}}$ is the projector on the support of the previous outcome operator, $B^{(1)}_{k_{1}}$, which may have a non-trivial kernel, so that in general it holds $\mathbf{1}_{k_{1}}\leq\mathbf{1}$. This problem may be overcome easily by redefining the POVM elements as $\tilde{B}^{(2)}_{k_{(1,2)}}=B^{(2)}_{k_{(1,2)}}\oplus\left(\mathbf{1}-\mathbf{1}_{k_{1}}\right)/2$, i.e., trivially expanding the support of those already defined in \eqref{second}, so that $\tilde{B}^{(2)}_{k_{1},0}+\tilde{B}^{(2)}_{k_{1},1}=\mathbf{1}_{k_{1}}\oplus\left(\mathbf{1}-\mathbf{1}_{k_{1}}\right)=\mathbf{1}$. This operation is trivial because, in the construction \eqref{nested} of the nested POVM, the operators $B^{(2)}_{k_{(1,2)}}$ always act after the operator $B^{(1)}_{k_{1}}$, so that the value of the former outside the support of the latter is completely irrelevant. In other words, completeness of the binary POVM's is not necessary for the definition of $\mathcal{N}^{(N)}$ as a proper POVM; it is sufficient to ask for \textit{weak completeness}, i.e., that $\mathcal{B}^{(u)}_{k_{(1,u-1)}}$ is complete on the support of the operator preceding it in the decomposition, $B^{(u-1)}_{k_{(1,u-1)}}$. \\
Generalizing the previous discussion, at the $u$-th step we can define the elements of the $2^{u-1}$ possible POVM's $\mathcal{B}^{(u)}_{k_{(1,u-1)}}$ as 
   \begin{eqnarray}\label{elements}
 &&B^{(u)}_{k_{(1,u)}}=\sqrt{B^{(u-1)}_{k_{(1,u-1)}}}^{-1}\cdots\sqrt{B^{(1)}_{k_{1}}}^{-1}\\
 &&\cdot\sum_{k_{(u+1,u_{F})}}E_{k_{(1,u)},k_{(u+1,u_{F})}}\sqrt{B^{(1)}_{k_{1}}}^{-1}\cdots\sqrt{B^{(u-1)}_{k_{(1,u-1)}}}^{-1}.\nonumber
 \end{eqnarray}
These elements are positive by construction and they satisfy the weak completeness relation $B^{(u)}_{k_{(1,u-1),0}}+B^{(u)}_{k_{(1,u-1),1}}=\mathbf{1}_{k_{(1,u-1)}}$, which is sufficient to define the POVM $\mathcal{N}^{(N)}$, as discussed in Appendix~\ref{appA}. Hence we are left to show that, when combining the binary elements Eq.~\eqref{elements} as in Eq.~\eqref{nested}, the elements $F_{k_{(1,u_{F})}}$ so constructed are equal to the $E_{k_{(1,u_{F})}}$. Indeed let us evaluate Eq.~\eqref{elements} for $u=u_{F}$, i.e., at the last step, noting that the sum contains only one term:
\begin{eqnarray}
 B^{(u_{F})}_{k_{(1,u_{F})}}&=&\sqrt{B^{(u_{F}-1)}_{k_{(1,u_{F}-1)}}}^{-1}\cdots\sqrt{B^{(1)}_{k_{1}}}^{-1}\nonumber\\
 &\cdot& E_{k_{(1,u_{F})}}\sqrt{B^{(1)}_{k_{1}}}^{-1}\cdots\sqrt{B^{(u_{F}-1)}_{k_{(1,u_{F}-1)}}}^{-1}.\label{last}
\end{eqnarray}
Let us then successively invert the outer square roots on the left-hand side of the equation exactly $u_{F}-1$ times, to obtain the relation 
\begin{eqnarray} 
E_{k_{(1,u_{F})}}=\left|\sqrt{B^{(u_{F})}_{k_{(1,u_{F})}}} \cdots\sqrt{B^{(1)}_{k_{1}}}\right|^{2}\equiv F_{k_{(1,u_{F})}},
\end{eqnarray}  which demonstrates that we can recover the initial POVM with the procedure outlined above. \\
This completes the proof when $N$ is an exact power of $2$. If this is not the case, it means that $\log_{2}N$ is not an integer and it suffices to consider the nested decomposition for the next higher integer, i.e., set $u_{F}=\lceil\log_{2}N\rceil+1$, $\tilde{N}=2^{u_{F}}$. Let us then trivially expand the initial $N$-outcome POVM to a $\tilde{N}$-outcome one as  
\begin{eqnarray} 
\mathcal{M}^{(\tilde{N})}=\mathcal{M}^{(N)}\cup\left\{E_{k_{(1,u_{F})}}=0, \forall j^{(k)}>N-1\right\},\end{eqnarray}  by adding $\tilde{N}-N$ null elements. The nested decomposition $\mathcal{N}^{(\tilde{N})}$ equivalent to $\mathcal{M}^{(\tilde{N})}$ can be computed again by Eqs.~(\ref{nested},\ref{elements}) and it comprises $\tilde{N}-N$ null elements too. If we isolate these elements from the rest we obtain a decomposition \begin{eqnarray} \mathcal{N}^{(\tilde{N})}=\mathcal{N}^{(N)}\cup\left\{F_{k_{(1,u_{F})}}=0, \forall j^{(k)}>N-1\right\},\end{eqnarray}  where $\mathcal{N}^{(N)}$ can be interpreted as a nested representation of the initial POVM $\mathcal{M}^{(N)}$.  \end{proof}
 
\section{An application: optimal quantum state discrimination}\label{StateDisc}
In this Section we apply the previous POVM decomposition to the problem of optimal state discrimination. Let us suppose we are given one copy of a quantum state, represented by a positive and trace-one operator $\rho_{j}$, chosen at random from a set $\mathcal{S}^{(N)}=\{\tilde{\rho}_{j}=p_{j}\rho_{j}\}_{j=0,\cdots,N-1}$ of $N$ states weighted with probability $p_{j}$, so that $\sum_{j=0}^{N-1}p_{j}=1$; we have to perform a measurement to decide which state was sent. If the states are not orthogonal, i.e., $\rho_{j}\rho_{k}\neq0$ for some values of $j,k$, and we are constrained to give a conclusive answer, there exists no measurement that can succeed with unit probability. The average success probability of discriminating the set of states $S^{(N)}$ with a $N$-outcome POVM $\mathcal{M}^{(N)}$, as defined in Sec.~\ref{Deco}, can be computed as 
\begin{eqnarray} P_{Succ}\left(S^{(N)},\mathcal{M}^{(N)}\right)=\sum_{j=0}^{N-1}\operatorname{Tr}\left[E_{j}\tilde{\rho}_{j}\right],\end{eqnarray} 
where each measurement outcome $E_{j}$ is associated with the detection of the respective weighted state $\tilde{\rho}_{j}$.
We are particularly interested in the optimal success probability, obtained by optimizing over all measurements: \begin{eqnarray} \mathbb{P}_{Succ}(\mathcal{S}^{(N)})=\max_{\mathcal{M}^{(N)}}P_{Succ}\left(\mathcal{S}^{(N)},\mathcal{M}^{(N)}\right).\end{eqnarray}  \\
Following Sec.~\ref{Deco}, we can always decompose the discrimination measurement into a sequence of nested binary ones, writing the success probability as \begin{eqnarray}
&&P_{Succ}\left(\mathcal{S}^{(N)},\mathcal{N}^{(N)}\right)=\sum_{k_{(1,u_{F})}}\operatorname{Tr}\left[F_{k_{(1,u_{F})}}\tilde{\rho}_{k_{(1,u_{F})}}\right]\nonumber\\
&&=\sum_{k_{(1,u_{F})}}\operatorname{Tr}\left[\left|\sqrt{B^{(u_{F})}_{k_{(1,u_{F})}}} \cdots\sqrt{B^{(1)}_{k_{1}}}\right|^{2}\tilde{\rho}_{k_{(1,u_{F})}}\right],\label{probNest}
\end{eqnarray}
where we have introduced the binary representation $k_{(1,u_{F})}$ for the labels $j$ of the states and measurement operators and employed the definition \eqref{nested} for the elements of the nested POVM.
This decomposition is interesting because it establishes a relation between the discrimination probability of a given set of states and that of its subsets of smaller size. Let us indeed suppose that the first measurement is successful, i.e., that an outcome $k_{1}$ occurs if one of the states $\rho_{k_{1},k_{(2,u_{F})}}$ with that value of the first bit was present. This happens with probability  $p_{Succ}(k_{1})=\sum_{k_{(2,u_{F})}}\operatorname{Tr}\left[B^{(1)}_{k_{1}}\tilde{\rho}_{k_{1},k_{(2,u_{F})}}\right]$. In this case the possible weighted states after the measurement are \begin{eqnarray} \tilde{\tau}_{k_{1},k_{(2,u_{F})}}=\sqrt{B^{(1)}_{k_{1}}}\tilde{\rho}_{k_{1},k_{(2,u_{F})}}\sqrt{B^{(1)}_{k_{1}}}/p_{Succ}(k_{1}),\end{eqnarray} forming a set of size $N/2$: $\mathcal{S}^{(N/2)}_{k_{1}}=\left\{\tilde{\tau}_{k_{1},k_{(2,u_{F})}}\right\}_{k_{2},\cdots,k_{u_{F}}=0,1}$. Moreover the collection of remaining measurements can be seen as a nested POVM of order $N/2$: \begin{eqnarray} \mathcal{N}^{(N/2)}_{k_{1}}=\left\{\mathcal{B}^{(u_{F})}_{k_{1},k_{(2,u_{F}-1)}}\right\}_{k_{2},\cdots,k_{u_{F}-1}=0,1}\circ\cdots\circ\left\{\mathcal{B}^{(2)}_{k_{1}}\right\}. \nonumber \end{eqnarray}
Hence we can easily rewrite the probability of discriminating the set of states $\mathcal{S}^{(N)}$ with the POVM $\mathcal{N}^{(N)}$, Eq.~\eqref{probNest}, as the probability of discriminating the set $\mathcal{S}^{(N/2)}_{k_{1}}$ with the POVM $\mathcal{N}^{(N/2)}_{k_{1}}$ if the first measurement had an outcome $k_{1}$, averaged over all values of $k_{1}=0,1$:
\begin{eqnarray}
P_{Succ}\Big(\mathcal{S}^{(N)}&&,  \mathcal{N}^{(N)}\Big)=\sum_{k_{1},k_{(2,u_{F})}}p_{Succ}(k_{1})\nonumber\\ &&\cdot\operatorname{Tr}\Bigg[\left|\sqrt{B^{(u_{F})}_{k_{1},k_{(2,u_{F})}}}\cdots\sqrt{B^{(2)}_{k_{1},k_{2}}}\right|^{2}\tilde{\tau}_{k_{1},k_{(2,u_{F})}}\Bigg]\nonumber\\
&&=\sum_{k_{1}}p_{Succ}(k_{1})P_{Succ}\left(\mathcal{S}^{(N/2)}_{k_{1}},\mathcal{N}^{(N/2)}_{k_{1}}\right).\label{rec}
\end{eqnarray}
This expression suggests a recursive optimization: if the optimal discrimination problem is solved for any set  $\mathcal{S}^{(N/2)}$ of a fixed size $N/2$, possibly restricting to a specific Hilbert space, e.g., qubits or continuous-variable states, then the result can be plugged into \eqref{rec} to obtain an expression for the discrimination probability of a set of double size, which depends on a single couple of measurement operators, i.e., the first-step binary POVM $\mathcal{B}^{(1)}$. However, if a general solution for the discrimination of a smaller set of states is not available, the problem remains hard, since when optimizing $P_{Succ}\left(\mathcal{S}^{(N/2)}_{k_{1}},\mathcal{N}^{(N/2)}_{k_{1}}\right)$ one still has to take into account the dependence of the states of $\mathcal{S}^{(N/2)}_{k_{1}}$ on the first-step POVM, which is itself subject to optimization afterwards, thus making the states arbitrary.\\
Fortunately the first step of the recursion has a well-known solution~\cite{Hel}: 
\begin{eqnarray} \mathbb{P}_{Succ}(\mathcal{S}^{(2)})=\left(1+\norm{\tilde{\rho}_{0}-\tilde{\rho}_{1}}_{1}\right)/2,
\end{eqnarray} 
 where $\norm{\cdot}_{1}=\operatorname{Tr}[|\cdot|]$ is the trace norm of the argument. Then by plugging this expression into the optimization of Eq.~\eqref{rec} for $N=4$ states we can write:
\begin{eqnarray}
\mathbb{P}_{Succ}\left(\mathcal{S}^{(4)}\right)&=&\max_{\mathcal{B}^{(1)}}\sum_{k_{1}}p_{Succ}(k_{1})\frac{1+\norm{\tilde{\tau}_{k_{1},0}-\tilde{\tau}_{k_{1},1}}_{1}}{2}\nonumber\\
&=&\max_{\mathcal{B}^{(1)}}\sum_{k_{1}}\frac{1}{2}\Bigg(\operatorname{Tr}\left[B^{(1)}_{k_{1}}\left(\tilde{\rho}_{k_{1},0}+\tilde{\rho}_{k_{1},1}\right)\right]\nonumber\\
&+&\norm{\sqrt{B^{(1)}_{k_{1}}}\left(\tilde{\rho}_{k_{1},0}-\tilde{\rho}_{k_{1},1}\right)\sqrt{B^{(1)}_{k_{1}}}}_{1}\Bigg).
\end{eqnarray}
We can write the latter equation more compactly by introducing the function
\begin{eqnarray} \label{DEFFQABC}
\mathcal{F}_{Q}(A,B,C)&=&\operatorname{Tr}\left[QA+\left|\sqrt{Q}B\sqrt{Q}\right| \right.\nonumber \\
&&\qquad \left.+\left|\sqrt{\mathbf{1}-Q} C \sqrt{\mathbf{1}-Q}\right|\right],\end{eqnarray} 
where $Q$ is a positive and less-than-one operator, while the arguments $A,B,C$ are hermitian operators, and its maximum over $Q$, i.e., 
\begin{eqnarray} 
\mathcal{F}(A,B,C)=\max_{\mathbf{1}\geq Q\geq0}\mathcal{F}_{Q}(A,B,C). \label{DEFFABC}
\end{eqnarray}  Setting $B^{(1)}_{0}=Q$ and $B^{(1)}_{1}=\mathbf{1}-Q$, we obtain:
 \begin{eqnarray}\mathbb{P}_{Succ}\Big(\mathcal{S}^{(4)}\Big)=\frac{p_{1,0}+p_{1,1}}{2}+\mathcal{F}\left(A^{(4)},B^{(4)},C^{(4)}\right),\hspace{20pt}\label{prob4}\end{eqnarray} 
 with $A^{(4)}=(\tilde{\rho}_{0,0}+\tilde{\rho}_{0,1}-\tilde{\rho}_{1,0}-\tilde{\rho}_{1,1})/2$, $B^{(4)}=(\tilde{\rho}_{0,0}-\tilde{\rho}_{0,1})/2$ and $C^{(4)}=(\tilde{\rho}_{1,0}-\tilde{\rho}_{1,1})/2$. Similarly for $N=3$ states we have:
\begin{eqnarray}
\mathbb{P}_{Succ}\Big(\mathcal{S}^{(3)}\Big)=p_{1,0}+\mathcal{F}\left(A^{(3)},B^{(3)},C^{(3)}\right),\label{prob3}
\end{eqnarray}
with $A^{(3)}=(\tilde{\rho}_{0,0}+\tilde{\rho}_{0,1})/2-\tilde{\rho}_{1,0}$, $B^{(3)}=B^{(4)}$ as before and $C^{(3)}=0$.
Thus the optimal discrimination problem of $N=3,4$ states has been reduced to the evaluation of the function $\mathcal{F}$, which requires an optimization over a single operator $Q$.\\
As already discussed, if the problem of Eq.~\eqref{prob4} were to be solved exactly for any set of states $\mathcal{S}^{(4)}$, then the result could be plugged into Eq.~\eqref{rec}, obtaining an expression for the optimal discrimination probability of $N=8$ states dependent only on the first binary POVM. Unfortunately a solution of Eqs.~(\ref{prob4},\ref{prob3}) can be found only in some specific cases, listed below and discussed in detail in Appendix~\ref{appB}. In the following we employ the positive part of an operator $X$, defined as $X_{+}=(X+|X|)/2$.
 \newtheorem{proposition}{Proposition}
\begin{proposition}\label{cases}
The value of the function $\mathcal{F}(A,B,C)$ of Eq.~(\ref{DEFFABC}) is 
\begin{eqnarray}\label{value} \mathcal{F}\left(A,B,C\right)=\operatorname{Tr}[(A+|B|-|C|)_{+}]+\norm{C}_{1}, \end{eqnarray} 
when at least one of the following conditions holds: i) the operators $B$ and $C$ have support respectively on the positive and negative support of $A$; ii) $B$ and $C$ have a definite sign; iii) $A$, $B$ and $C$ all commute with each other.
 \end{proposition}
 \newtheorem{remark}{Remark}
 \begin{remark}
 In the first case of Proposition~\ref{cases}, i.e., that the operators $B$ and $C$ have support respectively on the positive and negative support of $A$, the expression \eqref{value} can be simplified as
 \begin{eqnarray}
 \mathcal{F}\left(A,B,C\right)=\operatorname{Tr}\left[A_{+}\right]+\norm{B}_{1}+\norm{C}_{1}.
 \end{eqnarray}
 \end{remark}
\begin{remark}\label{recRem}
The optimal success probability is invariant under exchange of the states, i.e., under relabelling of the indices $k_{1},k_{2}$ in our case $N=3,4$. Hence it can happen that the conditions listed in Proposition~\ref{cases} are valid only for $A, B, C$ given by a specific ordering of the states.
\end{remark}
The previous remark implies that, when checking whether a set of states satisfies the conditions of Proposition~\ref{cases} or not, one has to consider all possible sets of $A, B, C$ obtainable by different orderings of the states, not only the conventional one employed in Eqs.~(\ref{prob4},\ref{prob3}). Alternatively, one can apply this symmetry under exchange of the states to obtain recursive relations for $\mathcal{F}(A,B,C)$, e.g., for $N=3$ and by exchanging $(0,0)\leftrightarrow(1,0)$, it holds
\begin{equation}
\mathcal{F}(A,B,0)=\frac{\mathcal{F}(-3B-A,B-A,0)}{2}+\operatorname{Tr}[A+B];\label{rec3}
\end{equation}
then Proposition~\ref{cases} holds on the right-hand side of \eqref{rec3} when $B'=B-A$ has a definite sign, but the latter is simply $B'=(\tilde{\rho}_{1,0}-\tilde{\rho}_{0,1})/2$, an expression of the operator $B$ for the new ordering of the states. 
 \begin{remark}
 In all the cases listed in Proposition~\ref{cases}, with the conventional ordering of the states of Eqs.~(\ref{prob4},\ref{prob3}), the optimal success probabilities for the discrimination of $N=3,4$ states become\begin{widetext}
 \begin{eqnarray}
 &&\mathbb{P}_{Succ}\Big(\mathcal{S}^{(3)}\Big)=p_{1,0}+\operatorname{Tr}\left[\frac{\left(\tilde{\rho}_{0,0}+\tilde{\rho}_{0,1}-2\tilde{\rho}_{1,0}+|\tilde{\rho}_{0,0}-\tilde{\rho}_{0,1}|\right)_{+}}{2}\right],\label{p3}\\
 \label{p4}
 &&\mathbb{P}_{Succ}\Big(\mathcal{S}^{(4)}\Big)=\frac{p_{1,0}+p_{1,1}}{2}+\operatorname{Tr}\left[\frac{\left(\tilde{\rho}_{0,0}+\tilde{\rho}_{0,1}-\tilde{\rho}_{1,0}-\tilde{\rho}_{1,1}+|\tilde{\rho}_{0,0}-\tilde{\rho}_{0,1}|-|\tilde{\rho}_{1,0}-\tilde{\rho}_{1,1}|\right)_{+}}{2}\right]+\norm{\frac{\tilde{\rho}_{1,0}-\tilde{\rho}_{1,1}}{2}}_{1}.
 \end{eqnarray}\end{widetext} 
 \end{remark}

\section{A numerical example: the $N=3,4$ qubit case}\label{Qubits}
In this Section we analyze the discrimination probability obtained with the nested POVM decomposition in the case of $N=3,4$ \textit{qubit} states. Indeed, since Eqs.~(\ref{prob4},\ref{prob3}) seem not to be solvable analytically for generic sets of states, it is interesting to tackle the problem by choosing the simplest possible Hilbert space for the measured system, i.e., the qubit space $\mathcal{H}_{2}$ of dimension two. It is well known that the density matrices $\rho$ of this system can be represented 
 as a real  vector $\vec{v}_{\rho}$ inside a three-dimensional unit sphere (the Bloch sphere), i.e. 
  $\rho=(\mathbf{1}_{2}+\vec{v}_{\rho}\cdot\vec{\sigma})/2$ where $\mathbf{1}_{2}$ is the identity operator and 
  $\vec{\sigma} = (\sigma_1,\sigma_2,\sigma_3)$ is the vector of Pauli matrices 
\begin{eqnarray} \sigma_{1}=\left(\begin{array}{cc}0&1\\1&0\end{array}\right), \;  \sigma_{2}=\left(\begin{array}{cc}0&-i\\i&0\end{array}\right),\; \sigma_{3}=\left(\begin{array}{cc}1&0\\0&-1\end{array}\right). \nonumber 
\end{eqnarray} 
  In particular, pure states are situated on the sphere's surface, i.e., $v_{\rho}=|\vec{v}_{\rho}|=1$ for $\rho=|\psi\rangle\langle\psi|$, while the completely mixed state $\mathbf{1}_{2}/2$ is at the origin. 
 More generally, any hermitian operator $X$ on the qubit space can be expressed in terms  of  four real coefficients: a scalar $c_{X}$, which represents the normalization coefficient of the operator, and a vector $\vec{r}_{X}$, which represents the operator in the Bloch space, i.e. 
 \begin{eqnarray} X=c_{X}\mathbf{1}_{2}+\vec{r}_{X}\cdot\vec{\sigma},\end{eqnarray} 
the trace of the operator being determined by  $\operatorname{Tr}[X]=2c_{X}$, while its eigenvalues by $\lambda_{X}^{(\pm)}=c_{X}\pm r_{X}$ with $r_X= |\vec{r}_X|$.

Employing the representation described above we can hence rewrite  the function $\mathcal{F}_{Q}(A,B,C)$ as (see Appendix~\ref{appC} for details) 
\begin{widetext}
\begin{eqnarray}
\mathcal{F}_{Q}(A,B,C)\Big|_{\mathcal{H}_{2}}&=&2\Bigg(c_{Q}c_{A}+\vec{r}_{Q}\cdot\vec{r}_{A}
+\sqrt{(c_{Q}c_{B}+\vec{r}_{Q}\cdot\vec{r}_{B})^{2}+\left(\left(r_{B}\right)^{2}-\left(c_{B}\right)^{2}\right)\left(\left(c_{Q}\right)^{2}-\left(r_{Q}\right)^{2}\right)}\nonumber\\
&+&\sqrt{((1-c_{Q})c_{C}-\vec{r}_{Q}\cdot\vec{r}_{C})^{2}+\left(\left(r_{B}\right)^{2}-\left(c_{B}\right)^{2}\right)\left(\left(1-c_{Q}\right)^{2}-\left(r_{Q}\right)^{2}\right)}\Bigg),\label{fQub}
\end{eqnarray}
 when $B$ and $C$ do not have a definite sign, or
\begin{eqnarray}
\mathcal{F}_{Q}(A,B,C)\Big|_{\mathcal{H}_{2}}&=&2\Bigg(c_{Q}c_{A}+\vec{r}_{Q}\cdot\vec{r}_{A}+c_{Q}c_{|B|}+\vec{r}_{Q}\cdot\vec{r}_{|B|}+\left(1-c_{Q}\right)c_{|C|}-\vec{r}_{Q}\cdot\vec{r}_{|C|}\Bigg)\nonumber\\
&=&2\left(c_{Q}c_{A+|B|-|C|}+\vec{r}_{Q}\cdot\vec{r}_{A+|B|-|C|}+c_{|C|}\right),\label{defSign}
\end{eqnarray}\end{widetext} 
when both $B$ and $C$ have a definite sign, with $c_{|B|}=\pm c_{B}$, $\vec{r}_{|B|}=\pm\vec{r}_{B}$ respectively for $B\geq 0$ and $B\leq0$ and similar definitions for $|C|$.\\
For each set of $N=3,4$ qubit states to discriminate, the operators $A, B, C$, i.e., their coefficients $c$ and $\vec{r}$, are fixed and the optimization of Eqs.~(\ref{fQub}, \ref{defSign}) is to be carried out only over Q, i.e., on its coefficients $c_{Q}$ and $\vec{r}_{Q}$ under the constraints
\begin{eqnarray} 1\geq c_{Q}\geq0, \qquad r_{Q}\leq\min[c_{Q},1-c_{Q}]\;, \end{eqnarray} 
that ensure the positivity of $Q$ and the fact that it must be smaller than one.
In particular for $N=3$ states $C=0$ and one can show that $c_{Q}+r_{Q}=\lambda^{(+)}_{Q}=1$ is optimal. Moreover the optimal $\vec{r}_{Q}$ lies on the plane of $\vec{r}_{A}$ and $\vec{r}_{B}$, so that it can be defined in terms of its norm $r_{Q}$ and a single angle $\phi_{Q}$ as $\vec{r}_{Q}\cdot\vec{r}_{A}=r_{Q}r_{A}\cos\phi_{Q}$. Then in this case it is only required to optimize two parameters, namely $c_{Q}$ and $\phi_{Q}$. For $N=4$ instead there are no further simplifications and one has to optimize four parameters, with constraints.\\
 Let us now consider the case in which Eq.~\eqref{defSign} is valid, i.e., $B$ and $C$ have a definite sign. It is one of the situations considered in Proposition~\ref{cases}, thus the optimization of \eqref{defSign} must match the expression~\eqref{value}. This fact is already quite clear if we express Eq.~\eqref{defSign} directly in terms of the initial operators; furthermore it can also be shown by direct analytical optimization that in this case \begin{equation}\mathcal{F}(A,B,C)\Big|_{\mathcal{H}_{2}}=|c_{A+|B|-|C|}|+c_{A+|B|-|C|}+2c_{C}\end{equation} and the optimal value of $Q$ is $r_{Q}=0$ and $c_{Q}=\theta(c_{A+|B|-|C|})$, with $\theta(\cdot)$ the step function, valued $1$ when its argument is positive and zero otherwise. \\
 As for the other case, in which Eq.~\eqref{fQub} is valid, unfortunately the function cannot be completely optimized analytically. Nevertheless its numerical optimization is straightforward and thus, together with Eqs.~(\ref{prob4}, \ref{prob3}), it provides a convenient method to obtain the optimal success probability of discrimination and the optimal measurement operators for $N=3,4$ qubit states. 
\begin{figure}[t!]
\includegraphics[scale=.24]{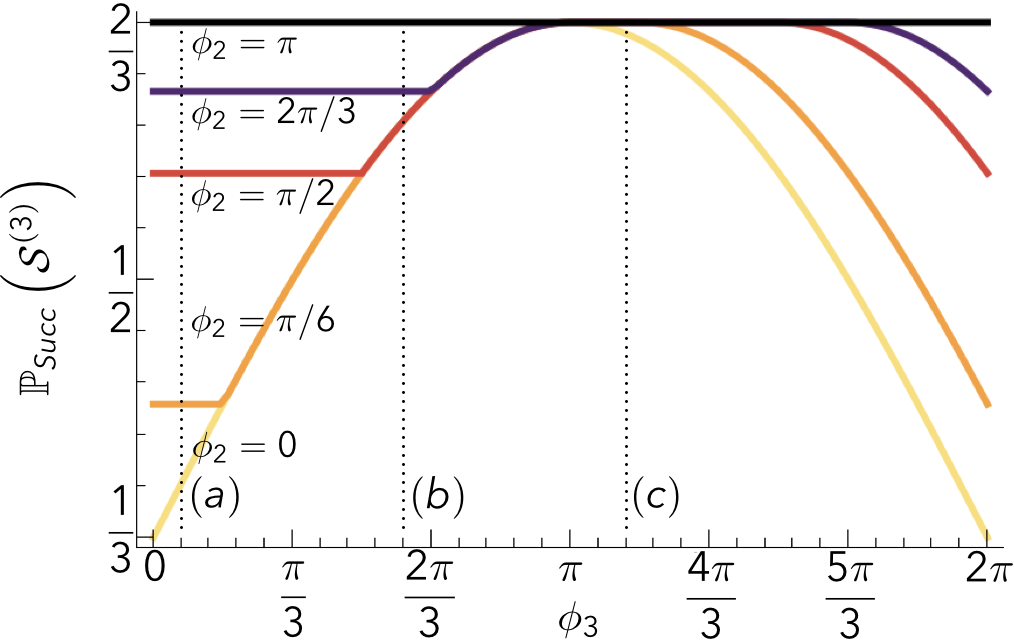}
\caption{Plot of $\mathbb{P}_{Succ}\left(\mathcal{S}^{(3)}\right)$, the optimal success probability \eqref{prob3} for a set of three equiprobable pure qubit states, identified by the Bloch vectors $\vec{r}_{\rho_{1}}=(1,0,0)$, $\vec{r}_{\rho_{2}}=(\cos\phi_{2},\sin\phi_{2},0)$ and $\vec{r}_{\rho_{3}}=(\cos\phi_{3},\sin\phi_{3},0)$, as a function of $\phi_{3}$ for several values of $\phi_{2}=0, \pi/6, \pi/2, 2\pi/3, \pi$ (respectively from yellow/light-gray to black). The results are obtained by numerical optimization of Eq.~\eqref{fQub} over $c_{Q}$, $\vec{r}_{Q}$. Observe that, for all values of $\phi_{2}$, there is a range of values of $\phi_{3}$ where $\mathbb{P}_{Succ}\left(\mathcal{S}^{(3)}\right)$ attains the maximum allowed for non-orthogonal states, i.e., the same as for symmetric states. Outside of this range the quantity decreases, reaching a constant minimum for $\phi_{3}\leq\phi_{2}$ (see text and Fig.~\ref{fig2} for an explanation). The cases $\phi_{2}=\pi/6, 2\pi/3$ are explicitly depicted in Fig.~\ref{fig2} for three values of $\phi_{3}$ identified by the labelled dotted lines.}\label{fig1}
\end{figure} 
As an example let us consider $N=3$ \textit{equiprobable pure} qubit states situated on the $(x,y)$ plane of the Bloch sphere, i.e., a combination of $\mathbf{1}_{2}$, $\sigma_{1}$ and $\sigma_{2}$; this is a simple choice for the sake of clarity, but we stress that no additional optimization difficulties are met when considering non-equiprobable and mixed states. Let us fix the first state to be on the $x$ axis, i.e., $\vec{r}_{\rho_{1}}=(1,0,0)$, without loss of generality. Then we can study the optimal success probability by varying the angles of the other two vectors with respect to the first one: $\vec{r}_{\rho_{2}}=(\cos\phi_{2},\sin\phi_{2},0)$ and $\vec{r}_{\rho_{3}}=(\cos\phi_{3},\sin\phi_{3},0)$. If the states are also symmetrically distributed at constant angles along the circumference, i.e., $\phi_{2}=\phi_{3}=2\pi/3$, the result is well-known~\cite{Hel}: $\mathbb{P}_{Succ}\left(\mathcal{S}_{sym}^{(3)}\right)=2/3$, which is the maximum success probability of discrimination for any three equiprobable qubit states (because no other configuration can achieve a lower average state-overlap than this one). 
\begin{figure}[t!]
\includegraphics[scale=.24]{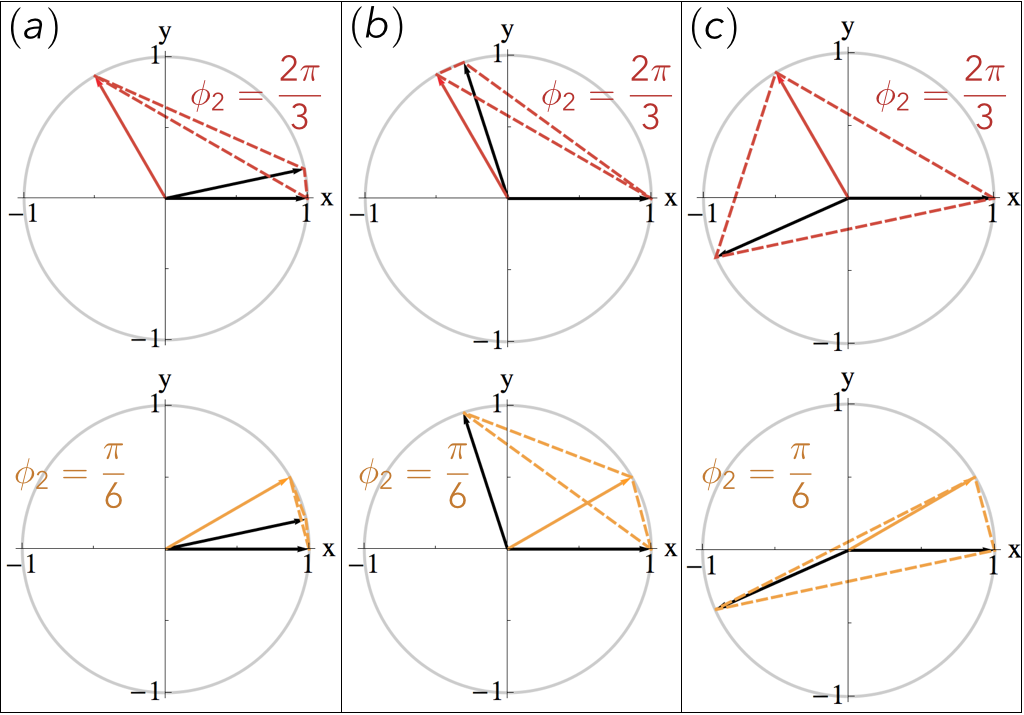}
\caption{Plot of the vectors of the three states $\vec{r}_{\rho_{1}}=(1,0,0)$ (black fixed on $x$ axis), $\vec{r}_{\rho_{3}}=(\cos\phi_{3},\sin\phi_{3},0)$ (black) and $\vec{r}_{\rho_{2}}=(\cos\phi_{2},\sin\phi_{2},0)$ (red/dark gray at the top and orange/light gray at the bottom) on the $(x,y)$ Bloch plane, as well as the triangles formed by them (same color codes as $\vec{r}_{\rho_{2}}$). The labels $a, b, c$ refer respectively to values of $\phi_{3}=\pi/15, 3\pi/5, 7\pi/5$, also highlighted in Fig.~\ref{fig1}, while two values of $\phi_{2}= 2\pi/3,\pi/6$ (respectively red/dark gray and orange/light gray figures) are considered. By comparison with Fig.~\ref{fig1} it is evident that $\mathbb{P}_{Succ}\left(\mathcal{S}^{(3)}\right)$ is maximum when the triangle formed by the states contains the origin ($c$), while it is lower otherwise. In particular, when $\phi_{3}\leq\phi_{2}$ ($a, b$ top, $a$ bottom) the largest side of the triangle formed by the states is always $\vec{r}_{\rho_{2}}-\vec{r}_{\rho_{1}}$ and this determines completely the optimal success probability.}\label{fig2}
\end{figure}
For more general states the results are shown in Fig.~\ref{fig1} where we plot the optimal success probability \eqref{prob3}, computed by numerical optimization of \eqref{fQub} over $c_{Q}$, $\vec{r}_{Q}$, as a function of the third angle $\phi_{3}$ and for several choices of $\phi_{2}=0, \pi/6, \pi/2, 2\pi/3, \pi$. It can be seen that, for all values of $\phi_{2}$, there is a range of values of $\phi_{3}$ for which $\mathbb{P}_{Succ}\left(\mathcal{S}^{(3)}\right)$ is equal to the maximum value of $2/3$, even though the states are not symmetrically distributed on the circumference. In other words there is a wide class of states that can be discriminated with performance as good as if they were symmetric. Outside of this range, whose width depends on $\phi_{2}$, the value of  $\mathbb{P}_{Succ}\left(\mathcal{S}^{(3)}\right)$ decreases and it reaches a constant minimum when $\phi_{3}\leq\phi_{2}$. \\
These peculiarities can be explained by referring to Refs.~\cite{qubits,bae,threeQubits}. In particular Ref.~\cite{bae} states that the optimal success probability of discrimination of a set $\mathcal{S}_{eq}^{(N)}=\{\rho_{j}/N\}_{j=0,\cdots,N-1}$ of $N$ \textit{equiprobable} qubit states can be found by: i) considering the geometric figure determined by the weighted states in the Bloch space, i.e., their polytope of vertices $\{\vec{r}_{\rho_{j}}/N\}$; ii) finding the polytope similar to the latter and that is also maximal in the Bloch sphere; iii) computing the ratio $R$ between the original and the maximal polytope. Then the optimal success probability is $\mathbb{P}_{Succ}\left(\mathcal{S}_{eq}^{(N)}\right)=\frac{1}{N}+R$. 
In light of this observation, we can explain the results of Fig.~\ref{fig1} by plotting the states in the Bloch sphere and studying the polygon formed by their vertices, as done in Fig.~\ref{fig2} for two values of $\phi_{2}=\pi/6, 2\pi/3$ used in Fig.~\ref{fig1} and $\phi_{3}=\pi/15, 3\pi/5, 7\pi/5$, corresponding to the dotted lines labelled $a, b, c$ in Fig.~\ref{fig1}. If the polytope determined by the qubits, usually a triangle, contains the origin, then the optimal success probability is still maximum, i.e., $\mathbb{P}_{Succ}\left(\mathcal{S}_{\supseteq0}^{(3)}\right)\equiv\mathbb{P}_{Succ}\left(\mathcal{S}_{sym}^{(3)}\right)$, as in Fig.~\ref{fig2}$c$. Indeed in this case the polytope formed by the states is already maximal in the Bloch sphere and $R=1/3$, as for the symmetric set. On the other hand, if the polytope does not contain the origin, the optimal success probability is strictly lower than the maximum one, i.e.,  $\mathbb{P}_{Succ}\left(\mathcal{S}_{\nsupseteq0}^{(3)}\right)<\mathbb{P}_{Succ}\left(\mathcal{S}_{sym}^{(3)}\right)$, as in Fig.~\ref{fig2}$a,b$. Indeed in this second case the polytope formed by the states is not maximal and can be expanded until its largest side matches a diameter of the circumference, so that $R<1/3$.
As for the region $\phi_{3}\leq\phi_{2}$ where $\mathbb{P}_{Succ}\left(\mathcal{S}_{\nsupseteq0}^{(3)}\right)$ is minimum and constant (as in Fig.~\ref{fig2}$a,b$ top and $a$ bottom), it can be explained by observing that the largest side of the triangle determined by the states, which in turn determines $R$, is always the one that connects $\vec{r}_{\rho_{1}}$ and $\vec{r}_{\rho_{2}}$, independently of $\vec{r}_{\rho_{3}}$. Since $\vec{r}_{\rho_{2}}-\vec{r}_{\rho_{1}}$ is constant for constant $\phi_{2}$, $R$ is constant too in this case.

\section{Conclusions} \label{Conc}
In this article we proposed a method to compute the optimal discrimination probability and optimal measurement operators of an arbitrary set of $N$ states. We showed how to decompose any multiple-outcome measurement into several binary-outcome steps, which could be of interest also in other contexts. For the discrimination problem this decomposition introduces a connection between the success probabilities of sets of different sizes, possibly simplifying the optimization procedure, but does not allow to reach a general analytical solution. Nevertheless it proves to be a useful tool for quickly determining the optimal discrimination probability of $N=3, 4$ qubit states, requiring just a simple numerical optimization. Indeed with this method we were able to highlight some interesting properties, explicitly verifying the validity and geometric insight of some previous results~\cite{qubits,bae,threeQubits}. Future lines of work could focus on simplifying the optimization for higher-dimensional systems and larger sets of states or investigating different kinds of measurement decompositions.

\section{Acknowledgments}
GdP acknowledges financial support from the European Research Council (ERC Grant Agreement no 337603), the Danish Council for Independent Research (Sapere Aude) and VILLUM FONDEN via the QMATH Centre of Excellence (Grant No. 10059).

\appendix
\section{Completeness of the nested POVM}\label{appA} 
In this appendix we show that the nested POVM defined in \eqref{nested} is complete, i.e., the sum of its elements is the identity on the whole Hilbert space of the system. This can be shown by employing the completeness of each binary POVM's $\mathcal{B}^{(u)}_{k_{(1,u-1)}}$ at each step $u$. Indeed we can start by summing over the last bit $k_{u_{F}}=0,1$, coupling elements that differ only for its value, i.e., $F_{k_{(1,u_{F}-1)},0}$ and $F_{k_{(1,u_{F}-1)},1}$. 
These are made of the same sequence of operators apart from the most interior ones, $B^{(u_{F})}_{k_{(1,u_{F}-1)},0}$ and $B^{(u_{F})}_{k_{(1,u_{F}-1)},1}$, which are instead two different elements of the same binary POVM $\mathcal{B}^{(u_{F})}_{k_{(1,u_{F}-1)}}$, thus satisfy a completeness relation and their sum can be simplified. The same procedure is then applied recursively on previous bits as follows:
\begin{widetext}
\begin{eqnarray}
\sum_{k_{1},\cdots,k_{u_{F}}}F_{k_{(1,u_{F})}}&=&\sum_{k_{1},\cdots,k_{u_{F}-1}}\left(F_{k_{(1,u_{F}-1)},0}+F_{k_{(1,u_{F}-1)},1}\right)\nonumber\\
&=&\sum_{k_{1},\cdots,k_{u_{F}-1}}\sqrt{B^{(1)}_{k_{1}}}\cdots\sqrt{B^{(u_{F}-1)}_{k_{(1,u_{F}-1)}}}\left(B^{(u_{F})}_{k_{(1,u_{F}-1)},0}+B^{(u_{F})}_{k_{(1,u_{F}-1)},1}\right)\sqrt{B^{(u_{F}-1)}_{k_{(1,u_{F}-1)}}}\cdots\sqrt{B^{(1)}_{k_{1}}}\nonumber\\
&=&\sum_{k_{1},\cdots,k_{u_{F}-1}}\sqrt{B^{(1)}_{k_{1}}}\cdots\sqrt{B^{(u_{F}-2)}_{k_{(1,u_{F}-2)}}}B^{(u_{F}-1)}_{k_{(1,u_{F}-1)}}\sqrt{B^{(u_{F}-2)}_{k_{(1,u_{F}-2)}}}\cdots\sqrt{B^{(1)}_{k_{1}}}\nonumber\\
&=&\sum_{k_{1},\cdots,k_{u_{F}-2}}\sqrt{B^{(1)}_{k_{1}}}\cdots\sqrt{B^{(u_{F}-2)}_{k_{(1,u_{F}-2)}}}\left(B^{(u_{F}-1)}_{k_{(1,u_{F}-2)},0}+B^{(u_{F}-1)}_{k_{(1,u_{F}-2)},1}\right)\sqrt{B^{(u_{F}-2)}_{k_{(1,u_{F}-2)}}}\cdots\sqrt{B^{(1)}_{k_{1}}}\nonumber\\
&=&\cdots=\sum_{k_{1}}\sqrt{B^{(1)}_{k_{1}}}\left(B^{(2)}_{k_{1},0}+B^{(2)}_{k_{1},1}\right) \sqrt{B^{(1)}_{k_{1}}} =B^{(1)}_{0}+B^{(1)}_{1}=\mathbf{1}.\label{complete}
\end{eqnarray}
\end{widetext}
We note that the previous result does not change if instead of employing complete binary POVM's, we relax to weak completeness, as defined in Sec.~\ref{Deco}, i.e., that each measurement $\mathcal{B}^{(u)}_{k_{(1,u-1)}}$ is complete on the support of the operator that preceeds it in the nested decomposition, $B^{(u-1)}_{k_{(1,u-1)}}$. In this case it still holds $\sqrt{B^{(u-1)}_{k_{(1,u-1)}}}\left(B^{(u)}_{k_{(1,u-1)},0}+B^{(u)}_{k_{(1,u-1)},1}\right)\sqrt{B^{(u-1)}_{k_{(1,u-1)}}}=B^{(u-1)}_{k_{(1,u-1)}}$ and the equalities in \eqref{complete} are unchanged.

\section{Detailed study of $\mathcal{F}$}\label{appB} 
In this appendix we study the function $\mathcal{F}_{Q}(A,B,C)$ appearing in Eqs.~(\ref{DEFFQABC}), and discuss its optimization in the cases mentioned in Sec.~\ref{StateDisc}. 
The optimization of $\mathcal{F}_{Q}(A,B,C)$ is difficult because of the competing interests of the three terms composing it. Indeed each single term of Eq.~\eqref{DEFFQABC} is maximized by a different operator $Q$: the first one is maximum when $Q$ is the projector on the positive support of $A$; the second one is maximum when $Q$ is the identity on the whole Hilbert space of the system; the third one is maximum when $Q$ is zero. Hence we can solve the problem exactly only if the three operators exhibit specific properties. \\
We start by observing that the function is subadditive in all its arguments, i.e., for any set of operators $\{A_{j}, B_{j}, C_{j}\}_{j=1,\cdots,n}$ it holds:
\begin{eqnarray}
\mathcal{F}(\sum_{j}A_{j},\sum_{j}B_{j},\sum_{j}C_{j})\leq\sum_{j}\mathcal{F}(A_{j},B_{j},C_{j}).\label{subAdd}
\end{eqnarray}
This follows from the subadditivity of the trace norm. 
We can now state some lemmas that help demonstrate Proposition~\ref{cases}. Throughout the Appendix, the notation $\mathbf{1}_{X}$ represents the projector on the support of the operator $X$.
\newtheorem{lemma}{Lemma}
\begin{lemma}\label{aPos}
Let us suppose that $A$ is positive semidefinite, $B$ has support inside the support of $A$ and that $C$ and $A$ have orthogonal supports. Then $\mathcal{F}(A,B,C)=\operatorname{Tr}[A]+\norm{B}_{1}+\norm{C}_{1}$.
\end{lemma}
\begin{proof}
Consider the set of operators $\left\{A_{j}=A\delta_{j,1}, B_{j}=B\delta_{j,2}, C_{j}=C\delta_{j,3}\right\}_{j=1,2,3}$ and apply the subadditivity property \eqref{subAdd}: 
\begin{eqnarray}
\mathcal{F}(A,B,C)&\leq&\mathcal{F}(A,0,0)+\mathcal{F}(0,B,0)+\mathcal{F}(0,0,C)\nonumber\\
&=&\operatorname{Tr}[A]+\norm{B}_{1}+\norm{C}_{1}.\label{subAddApp}
\end{eqnarray}
The latter inequality can be saturated under the hypotheses of this lemma, by taking $Q=\mathbf{1}_{A}$.\end{proof}

\begin{lemma}\label{aNeg}
Let us suppose that $A$ is negative semidefinite, $C$ has support inside the support of $A$ and that $B$ and $A$ have orthogonal supports. Then $\mathcal{F}(A,B,C)= \norm{B}_{1}+\norm{C}_{1}$.
\end{lemma}
\begin{proof}
Consider the same set of operators of Lemma~\ref{aPos} and apply again the subadditivity property \eqref{subAddApp}, then use the fact that $A\leq0$: 
\begin{align}
\mathcal{F}(A,B,C)&\leq\mathcal{F}(A,0,0)+\mathcal{F}(0,B,0)+\mathcal{F}(0,0,C)\\
&=\norm{B}_{1}+\norm{C}_{1}.\nonumber
\end{align}
The latter inequality can be saturated under the hypotheses of this lemma, by taking $Q=\mathbf{1}_{B}$. 
\end{proof}
Hence we can prove the first case of Proposition~\ref{cases}: let $A=A_{+}\oplus(- A_{-})$ be the decomposition of $A$ in terms of its positive and negative parts, with $A_{\pm}\geq 0$, and suppose that $B$, $C$ have support respectively inside the support of $A_{+}$, $A_{-}$. Then consider the set of operators $\left\{A_{j}=(-1)^{j+1}A_{(-1)^{j+1}}, B_{j}=B\delta_{j,1}, C_{j}=C\delta_{j,2}\right\}_{j=1,2}$ and apply the subadditivity property \eqref{subAdd}, together with Lemmas~\ref{aPos}, \ref{aNeg}: 
\begin{eqnarray}
\mathcal{F}(A,B,C)&\leq&\mathcal{F}\left(A_{+},B,0\right)+\mathcal{F}\left(-A_{-},0,C\right)\nonumber\\
&=&\operatorname{Tr}\left[A_{+}\right]+\norm{B}_{1}+\norm{C}_{1},\label{spec}
\end{eqnarray}
which is saturated by a measurement operator $Q=\mathbf{1}_{A_{+}}$. 
This expression is equivalent to that given in \eqref{value} under the current hypotheses, indeed in this case it holds 
\begin{eqnarray}(A+|B|-|C|)_{+}&=&\left((A_{+}+|B|)\oplus(-A_{-}-|C|)\right)_{+}\nonumber\\
&=&A_{+}+|B|,
\end{eqnarray}
so that \eqref{spec} becomes
\begin{eqnarray} 
\mathcal{F}\left(A,B,C\right)&=&\operatorname{Tr}\left[A_{+}+|B|\right]+\norm{C}_{1}\nonumber\\
&=&\operatorname{Tr}\left[(A_{+}+|B|-|C|)_{+}\right]+\norm{C}_{1}.
\end{eqnarray} 
As for the second and third cases of Proposition~\ref{cases}, let us first note that 
\begin{eqnarray}\label{NQB}
\norm{\sqrt{Q}B\sqrt{Q}}_1 &\le& \norm{\sqrt{Q}B_+\sqrt{Q}}_1 + \norm{\sqrt{Q}B_-\sqrt{Q}}_1\nonumber\\
&=& \operatorname{Tr}[Q(B_+ + B_-)] = \operatorname{Tr}[Q|B|]\;,
\end{eqnarray}
where $B_{\pm}$ are the positive and negative parts of $B$ as defined above for $A$, and analogously
\begin{equation} \label{NQC}
\norm{\sqrt{\mathbf{1}-Q}C\sqrt{\mathbf{1}-Q}}_1 \le \operatorname{Tr}[(\mathbf{1}-Q)|C|]\;.
\end{equation}
We then have
\begin{eqnarray}\label{FQ2}
\mathcal{F}_{Q}(A,B,C) &\le& \operatorname{Tr}[Q(A+|B|-|C|)] + \norm{C}_1\label{F1}\\
&\le& \operatorname{Tr}[(A+|B|-|C|)_+] + \norm{C}_1\;.\label{F2}
\end{eqnarray}
The inequality~\eqref{F2}, is saturated by taking $Q$ equal to the projector onto the support of $(A+|B|-|C|)_+$. The inequalities~(\ref{NQB},\ref{NQC}) and hence~\eqref{F1} are saturated in both the second and third cases of Proposition~\ref{cases}, though for different reasons:
\begin{itemize}
\item If $B$ and $C$ have a definite sign, then it holds $B=B_{+}$ or $B=B_{-}$, so that Eq.~\eqref{NQB} is saturated and analogously \eqref{NQC};
\item If $A$, $B$, $C$ all commute with each other, then Eqs.~(\ref{NQB},\ref{NQC}) are saturated by any operator $Q$ which commutes with both $B$ and $C$. Eventually the choice $Q=\mathbf{1}_{(A+|B|-|C|)_{+}}$ necessary to saturate Eq.~\eqref{F2} satisfies this latter condition in the case considered.
\end{itemize}
Finally we note that the previous case of commuting operators, as well as further results, can also be derived by applying the symmetry property of the optimal success probabilities (\ref{prob4}, \ref{prob3}) to obtain recursive formulas, as discussed after Remark~\ref{recRem}, but still a full solution cannot be found in this way.

\section{Computation of $\mathcal{F}_{Q}$ in the qubit case}\label{appC}
In this Section we derive the results (\ref{fQub}, \ref{defSign}) explicitly. 
As a preliminary recall that, for any three vectors $\vec{a}, \vec{b}, \vec{c}\in\mathbb{R}^{3}$ and the Pauli matrices $\vec{\sigma}$ it holds: 
\begin{eqnarray}\label{pauliab}
&&\left(\vec{a}\cdot\vec{\sigma}\right)\left(\vec{b}\cdot\vec{\sigma}\right)=\left(\vec{a}\cdot\vec{b}\right)+i \left(\vec{a}\times\vec{b}\right)\cdot\vec{\sigma},\label{vec1}\\
&&\left(\vec{a}\times\left(\vec{b}\times\vec{c}\right)\right)=\left(\vec{a}\cdot\vec{c}\right)\vec{b}-\left(\vec{a}\cdot\vec{b}\right)\vec{c}.\label{vec2}
\end{eqnarray}
Moreover, given a positive operator $Q$ on $\mathcal{H}_{2}$, the coefficients $c_{\sqrt{Q}}$, $\vec{r}_{\sqrt{Q}}$ of its square root $\sqrt{Q}$ can be expressed in terms of its coefficients $c_{Q}$, $\vec{r}_{Q}$ as:
\begin{eqnarray}
\begin{cases}
c_{Q}=\left(c_{\sqrt{Q}}\right)^{2}+\left(r_{\sqrt{Q}}\right)^{2}\\
r_{Q}=2c_{\sqrt{Q}}r_{\sqrt{Q}}
\end{cases}\hspace{-18pt}\leftrightarrow
\begin{cases}
c_{\sqrt{Q}}=\frac{\sqrt{c_{Q}+r_{Q}}+\sqrt{c_{Q}-r_{Q}}}{2}\\
r_{\sqrt{Q}}=\frac{\sqrt{c_{Q}+r_{Q}}-\sqrt{c_{Q}-r_{Q}}}{2}.
\end{cases}\hspace{10pt}
\label{sqrtQ}
\end{eqnarray}
with $\vec{r}_{Q}\parallel\vec{r}_{\sqrt{Q}}$.
In order to evaluate $\mathcal{F}_{Q}(A,B,C)$ we can compute its first two terms, while the third one is similar to the second one. Let us start with the product $QA$: it is a generic operator with coefficients 
\begin{eqnarray}
c_{QA}&=&c_{Q}c_{A}+\vec{r}_{Q}\cdot\vec{r}_{A}\label{cqa}\\
\vec{r}_{QA}&=&c_{Q}\vec{r}_{A}+c_{A}\vec{r}_{Q}+i\left(\vec{r}_{Q}\times\vec{r}_{A}\right),\label{rqa}
\end{eqnarray}
computed by applying Eq.~\eqref{vec2}. Thus the first term of $\mathcal{F}_{Q}$ is simply $\operatorname{Tr}[QA]=2c_{QA}$. \\
As for the product $\sqrt{Q}B\sqrt{Q}$, its first coefficient is simple: $c_{\sqrt{Q}B\sqrt{Q}}=\operatorname{Tr}[\sqrt{Q}B\sqrt{Q}]/2=c_{QB}$, easily obtained by relabelling Eq.~\eqref{cqa}. The vector of coefficients instead is
\begin{eqnarray}\label{rSqBSq}
&&\vec{r}_{\sqrt{Q}B\sqrt{Q}}=c_{\sqrt{Q}B}\vec{r}_{\sqrt{Q}}+c_{\sqrt{Q}}\vec{r}_{\sqrt{Q}B}+i\left(\vec{r}_{\sqrt{Q}B}\times\vec{r}_{\sqrt{Q}}\right)\hspace{20pt}\\
&&=c_{B}\vec{r}_{Q}+2\left(\vec{r}_{\sqrt{Q}}\cdot\vec{r}_{B}\right)\vec{r}_{\sqrt{Q}}+\left(\left(c_{\sqrt{Q}}\right)^{2}-\left(r_{\sqrt{Q}}\right)^{2}\right)\vec{r}_{B},\nonumber
\end{eqnarray}
where we have first computed the product between $\sqrt{Q}B$ and $\sqrt{Q}$, then substituted the expression for the former by relabelling once again the product $QA$ and employed \eqref{sqrtQ}.\begin{widetext} We are interested in the absolute value of $\sqrt{Q}B\sqrt{Q}$, i.e., the sum of the absolute value of its eigenvalues $\lambda^{(\pm)}_{\sqrt{Q}B\sqrt{Q}}=c_{\sqrt{Q}B\sqrt{Q}}\pm r_{\sqrt{Q}B\sqrt{Q}}$. Hence the only dependence of the final expression on \eqref{rSqBSq} is through its norm:

\begin{eqnarray}
\left(r_{\sqrt{Q}B\sqrt{Q}}\right)^{2}&=&\left(c_{B}r_{Q}\right)^{2}+4(\vec{r}_{\sqrt{Q}}\cdot\vec{r}_{B})^{2}\left(c_{\sqrt{Q}}\right)^{2}+\left(\left(c_{\sqrt{Q}}\right)^{2}-\left(r_{\sqrt{Q}}\right)^{2}\right)^{2}\left(r_{B}\right)^{2}+2c_{B}\left(\left(c_{\sqrt{Q}}\right)^{2}+\left(r_{\sqrt{Q}}\right)^{2}\right)\left(\vec{r}_{Q}\cdot\vec{r}_{B}\right)\nonumber\\
&=&\left(c_{Q}c_{B}+\vec{r}_{Q}\cdot\vec{r}_{B}\right)^{2}+\left(\left(r_{B}\right)^{2}-\left(c_{B}\right)^{2}\right)\left(\left(c_{Q}\right)^{2}-\left(r_{Q}\right)^{2}\right),\label{rMod}
\end{eqnarray}
\end{widetext}
which we have simplified by employing the relations \eqref{sqrtQ}. Eventually we have to distinguish between two cases:
 \begin{itemize}
 \item If both $B$ and $C$ have definite sign then they can always be taken to be positive semidefinite, up to a relabeling $0\leftrightarrow 1$ of the second bits $k_{2}$ of the original states. Then we have 
 \begin{eqnarray}
 &&\norm{\sqrt{Q}B\sqrt{Q}}_{1}=\lambda^{(+)}_{\sqrt{Q}B\sqrt{Q}}+\lambda^{(-)}_{\sqrt{Q}B\sqrt{Q}}=2c_{\sqrt{Q}B\sqrt{Q}},\nonumber\\
 &&\norm{\sqrt{\mathbf{1}-Q}C\sqrt{\mathbf{1}-Q}}_{1}=2c_{\sqrt{\mathbf{1}-Q}C\sqrt{\mathbf{1}-Q}};
  \end{eqnarray}
 \item If instead $B$ and $C$ do not have a definite sign, then it must hold $\lambda^{(+)}_{\sqrt{Q}B\sqrt{Q}}\geq 0$ and $\lambda^{(-)}_{\sqrt{Q}B\sqrt{Q}}\leq 0$ and similar relations for $C$, so that 
  \begin{eqnarray}
 &&\norm{\sqrt{Q}B\sqrt{Q}}_{1}=\lambda^{(+)}_{\sqrt{Q}B\sqrt{Q}}-\lambda^{(-)}_{\sqrt{Q}B\sqrt{Q}}=2r_{\sqrt{Q}B\sqrt{Q}},\nonumber\\
 &&\norm{\sqrt{\mathbf{1}-Q}C\sqrt{\mathbf{1}-Q}}_{1}=2r_{\sqrt{\mathbf{1}-Q}C\sqrt{\mathbf{1}-Q}}.
  \end{eqnarray}
 \end{itemize}
 Note that the third term $\norm{\sqrt{\mathbf{1}-Q}C\sqrt{\mathbf{1}-Q}}_{1}$ can be expressed in terms of the coefficients of $Q$ by observing that $c_{\mathbf{1}-Q}=(1-c_{Q})$ and $\vec{r}_{\mathbf{1}-Q}=-\vec{r}_{Q}$.\\
We can conclude that for $B$ and $C$ of non-definite sign
\begin{eqnarray}
\mathcal{F}_{Q}(A,B,C)=2(c_{QA}+r_{\sqrt{Q}B\sqrt{Q}}+r_{\sqrt{\mathbf{1}-Q}C\sqrt{\mathbf{1}-Q}}),\hspace{25pt}
\end{eqnarray} 
while for $B$ and $C$ of definite sign
\begin{eqnarray}
\mathcal{F}_{Q}(A,B,C)=2\left(c_{QA}+c_{QB}+c_{QC}\right),
\end{eqnarray}
which give respectively Eqs.~(\ref{fQub}, \ref{defSign}) after inserting the values of the coefficients computed above, i.e., Eqs.~(\ref{cqa}, \ref{rMod}).

\end{document}